\documentclass{article}

\usepackage{microtype}
\usepackage{graphicx}
\usepackage{subfigure}
\usepackage{booktabs} 
\usepackage[table,xcdraw]{xcolor}

\usepackage{algorithm}
\usepackage{multirow}

\usepackage[accepted]{icml2025}

\usepackage{hyperref}
\usepackage{makecell}

\usepackage{amsmath}
\usepackage{amssymb}
\usepackage{mathtools}
\usepackage{amsthm}
\usepackage{url}
\usepackage[capitalize,noabbrev]{cleveref}

\theoremstyle{plain}
\newtheorem{theorem}{Theorem}[section]
\newtheorem{proposition}[theorem]{Proposition}
\newtheorem{lemma}[theorem]{Lemma}

\theoremstyle{definition}
\newtheorem{definition}[theorem]{Definition}
\newtheorem{assumption}[theorem]{Assumption}
\theoremstyle{remark}
\newtheorem{remark}[theorem]{Remark}
\usepackage[textsize=tiny]{todonotes}

\newcolumntype{M}[1]{>{\centering\arraybackslash}m{#1}}
\begin{document}

\twocolumn[
\icmltitle{Constrained Policy Optimization for Provably Fair Order Matching}




\begin{icmlauthorlist}
\icmlauthor{Zehua Cheng}{oxfordcs}
\icmlauthor{Zhipeng Wang}{mancs}
\icmlauthor{Wei Dai}{flock}
\icmlauthor{Wenhu Zhang}{hkust}
\icmlauthor{Vadzim Mahilny}{flock}
\icmlauthor{David Shi}{delulu}
\icmlauthor{Elena Jia}{delulu}
\icmlauthor{Jiahao Sun}{flock}
\end{icmlauthorlist}

\icmlaffiliation{flock}{FLock.io, London, UK}
\icmlaffiliation{oxfordcs}{Department of Computer Science, University of Oxford, Oxford, UK}
\icmlaffiliation{mancs}{Department of Computer Science, University of Manchester, Manchester, UK}
\icmlaffiliation{hkust}{Department of Computer Science and Engineering, Hong Kong University of Science and Technology, Hong Kong SAR}
\icmlaffiliation{delulu}{deluthium.ai}
\icmlcorrespondingauthor{Zehua Cheng, Elena Jia, Jiahao Sun}{zehuac@ieee.org, elena@deluthium.ai, sun@flock.io}

\icmlkeywords{Machine Learning, Order Match}
\vskip 0.3in
]

\printAffiliationsAndNotice{}  

\begin{abstract}
Automated matching engines execute millions of orders per session, yet systematic asymmetries in latency, order size, and market access compound into persistent execution disparities that erode participant trust. We formulate provably fair order matching as a Constrained Markov Decision Process and propose \textbf{CPO-FOAM} (Constrained Policy Optimization with Feedback-Optimized Adaptive Margins). An inner loop computes an analytic trust-region step on the Fisher information manifold; a PID-controlled outer loop dynamically tightens safety margins, suppressing the sawtooth oscillations endemic to Lagrangian methods under non-stationary dynamics. Group fairness (demographic parity, equalized odds) enters the CMDP cost vector while individual Lipschitz fairness is enforced deterministically via spectral normalization. We prove BIBO stability and that the integral term drives steady-state violations to zero. On LOBSTER NASDAQ data across six market regimes, CPO-FOAM recovers 95.9\% of unconstrained throughput at 2.5\% constraint violation frequency; on crypto-asset LOB data under MEV injection it captures 98.4\% of the reward envelope at 3.2\% CVF. The method scales sub-linearly to $M{=}8$ constraints, settles on-chain within one Ethereum block, and yields a $2.1\times$ reward improvement on Safety-Gymnasium~\cite{ji2023safety}, confirming domain-agnostic generalization.
\end{abstract}

\section{Introduction}

Automated trading infrastructure now mediates the majority of global financial transactions, executing millions of order-matching decisions per session across equities, derivatives, and digital asset exchanges. The design of these matching engines carries significant implications: even moderate, systematic asymmetries---arising from differential latency, heterogeneous order sizes, or privileged market access---compound over time into persistent execution disparities that erode participant trust and suppress market participation~\cite{aquilina2022quantifying}. As regulatory emphasis shifts toward auditable algorithmic governance, there is growing institutional demand for mechanism designs that internalize fairness as an explicit, enforceable constraint rather than an after-the-fact diagnostic.

We frame the formal problem as learning an order-matching policy that optimizes long-run market-quality objectives---effective spreads, order book depth, fill rates, and execution throughput---subject to hard group-level and individual-level fairness constraints measured at high frequency and subjected to continuous regulatory audit. This formulation naturally maps to a Constrained Markov Decision Process (CMDP), where the reward signal captures composite market quality and the cost vector encodes demographic parity gaps, equalized-odds-type conditional disparities, and Lipschitz individual fairness violations computed over rolling observation windows.

Existing algorithmic approaches encounter fundamental limitations in this setting. Rule-based mechanisms such as First-In-First-Out (FIFO) and Pro-Rata allocation are structurally incapable of jointly optimizing efficiency and fairness: our experiments show that FIFO achieves only 4,200 orders/s throughput with effective spreads of 2.45 basis points and 100\% constraint violation frequency (Table~\ref{tab:main_results}). Unconstrained deep reinforcement learning agents---specifically Proximal Policy Optimization (PPO)~\cite{schulman2017proximal}---reach theoretical efficiency ceilings but violate all fairness constraints. Among constrained methods, Lagrangian PPO and RCPO~\cite{tessler2018reward} suffer from dual-variable oscillations under non-stationary data, producing transient recovery lengths exceeding 145 optimization steps and large violation areas under the curve. Vanilla CPO~\cite{achiam2017constrained} provides per-update near-feasibility guarantees but struggles with noisy constraint estimates and the combinatorial structure of matching action spaces. Interior-Point Policy Optimization (IPO)~\cite{liu2020ipo} introduces barrier-coefficient tuning sensitivities without bounding transient violations. These limitations motivate the central question of this paper: \emph{can we design a fairness-sensitive matching algorithm that preserves per-update near-feasibility across multiple fairness constraints while making monotone progress on market quality, under stable multiplier dynamics and tractable computation?}

We answer this question affirmatively by introducing \textbf{CPO-FOAM} (Constrained Policy Optimization with Feedback-Optimized Adaptive Margins). On each iteration, an inner loop computes the optimal update direction by solving a low-dimensional dual problem over the Fisher information manifold, analytically recovering the trust-region-constrained reward-maximizing step. An outer loop governed by a discrete-time PID controller then dynamically adjusts safety margins $\xi_{i,k}$ to buffer the linearized constraint boundary against stochastic estimation noise, proactively tightening the feasible set before violations materialize. Individual Lipschitz fairness is enforced deterministically via per-layer spectral normalization, decoupling architectural regularity from the probabilistic CMDP constraints. We prove that the closed-loop system satisfies the Jury stability criterion and that the integrator term drives steady-state violations to zero via the Final Value Theorem.

Our contributions are as follows:
\begin{enumerate}\setlength{\itemsep}{-0.5ex}
    \item \textbf{CMDP Formulation for Auditable Fair Matching.} We formulate fair order matching as a discounted CMDP whose cost vector encodes demographic-parity gaps, equalized-odds-type conditional disparities, and Lipschitz individual fairness violations computed over rolling windows, providing an auditable interface between microstructure telemetry and enforceable fairness constraints.
    \item \textbf{CPO-FOAM Algorithm.} We construct a two-stage update combining a trust-region improvement step with a KL-nearest fairness projection and PID-controlled adaptive safety margins. We prove BIBO stability of the closed-loop system and asymptotic feasibility under bounded stochastic disturbances, with a pure recovery step guaranteeing descent on the cumulative violation energy when the feasible set is empty.
    \item \textbf{Comprehensive Empirical Validation.} We provide a reproducible evaluation stack spanning LOBSTER NASDAQ reconstructions across six non-stationary market regimes, hardware scalability profiling up to $K=500$ action dimensions and $M=8$ simultaneous constraints, on-chain output-verifiable Ethereum settlement with gas cost analysis, and domain-agnostic generalization on Safety-Gymnasium~\cite{ji2023safety}. CPO-FOAM achieves strict Pareto dominance: 0.95 bps spread, 2.5\% CVF, violation transient length of 5 steps, and $2.1\times$ reward improvement on SafetyAntVelocity over the nearest baseline.
\end{enumerate}

\section{Related Work}

\subsection{Constrained Reinforcement Learning}

Constrained Markov Decision Processes (CMDPs) provide the foundational framework for optimizing sequential objectives under inequality constraints~\cite{altman1999constrained}. Constrained Policy Optimization (CPO)~\cite{achiam2017constrained} performs trust-region updates with theoretical guarantees for near-constraint satisfaction at each iterate by solving a linearized-quadratic subproblem over the Fisher information manifold. However, CPO assumes accurate constraint estimates and encounters practical difficulties with noisy cost signals and combinatorial action spaces. Reward-Constrained Policy Optimization (RCPO)~\cite{tessler2018reward} reformulates constraints as multi-timescale Lagrangian penalties, achieving computational simplicity at the cost of oscillatory multiplier dynamics and delayed convergence to feasibility. Interior-Point Policy Optimization (IPO)~\cite{liu2020ipo} translates inequality constraints into log-barrier penalties but introduces barrier-coefficient tuning sensitivities and provides no guarantees on transient violation magnitude. To address the instability of integral-only multiplier updates, PID-Lagrangian methods~\cite{stooke2020responsive} augment the dual variable dynamics with proportional and derivative terms drawn from classical control theory, substantially dampening oscillations on benchmark suites such as Safety-Gymnasium~\cite{ji2023safety}. Our approach unifies the geometric projection guarantees of CPO with the anticipatory stabilization of PID control, achieving the convergence benefits of both paradigms.

\subsection{Algorithmic Fairness in Sequential Decision-Making}

Fairness criteria originally developed for supervised classification---including demographic parity~\cite{calders2009building}, equalized odds~\cite{hardt2016equality}, and individual (Lipschitz) fairness~\cite{dwork2012fairness}---have been increasingly adapted to sequential settings. Recent work extends group fairness constraints to contextual bandits and Markov decision processes, typically encoding fairness as trajectory-level cost constraints within CMDP formulations~\cite{wen2021algorithms}. Individual fairness, formalized as Lipschitz continuity of the decision mapping, has been enforced via spectral normalization of neural network parameters~\cite{miyato2018spectral}. However, directly applying these criteria to high-frequency financial matching introduces unique challenges: execution outcomes are discrete, constraint estimates are autocorrelated rather than i.i.d., and the combinatorial structure of multi-counterparty allocations resists standard convex relaxations. Our work bridges this gap by deriving differentiable fairness surrogates compatible with gradient-based policy optimization and cleanly separating group-level rate constraints (handled by the CMDP projection) from individual-level regularity constraints (enforced architecturally via spectral normalization).

\subsection{Market Microstructure and Mechanism Design}

The design of electronic matching engines has been studied extensively within market microstructure theory~\cite{gould2013limit}. Traditional rule-based mechanisms---FIFO price-time priority, pro-rata allocation, and size-time interpolation---are deployed across major global exchanges but lack the capacity to jointly optimize multiple market quality dimensions or adapt to non-stationary order flow. Recent proposals such as frequent batch auctions~\cite{budish2015high} and speed bumps address latency arbitrage but do not formalize multi-dimensional fairness constraints. Machine learning approaches to order execution and market making have leveraged deep reinforcement learning for optimal placement~\cite{nevmyvaka2006reinforcement} and inventory management, but typically treat fairness as an exogenous reporting metric rather than an endogenous constraint. The JAX-LOB simulator~\cite{frey2023jax} provides a GPU-accelerated limit order book environment enabling the large-sample policy optimization and ablation studies required by our algorithm. Our work extends this line by embedding the matching engine within a CMDP framework that simultaneously optimizes market quality, enforces multi-dimensional fairness, and provides formal convergence guarantees.

\subsection{Verifiable Machine Learning on Distributed Ledgers}

The deployment of learned models in trustless decentralized environments demands cryptographic verification of both outputs and processes. Output-verifiable computation frameworks submit settlement commitments (state hashes, allocation Merkle roots, fairness metrics) alongside challenge-response protocols that enable any participant to dispute inconsistencies via on-chain adjudication~\cite{teutsch2019scalable}. Process-verifiable approaches---including open replay with model weights stored on IPFS~\cite{benet2014ipfs}/Arweave~\cite{williams2019arweave} and permissioned replay via threshold key-sharing committees with BLS multi-signatures~\cite{boneh2001short}---provide stronger provenance guarantees at the cost of model confidentiality or trusted verifier assumptions. Our on-chain proof system integrates both paradigms, demonstrating that CPO-FOAM settlements clear within standard Ethereum finality windows ($\leq$12 seconds) at economically viable gas costs ($\leq$725k per batch), with false challenge rates driven below 0.001\% by rapid off-chain recomputation.

\section{Preliminaries}
We model the Limit Order Book (LOB) matching problem as a CMDP, defined by the tuple $\mathcal{M} = (\mathcal{S}, \mathcal{A}, P, r, \mathbf{c}, \gamma, \mathbf{d})$, where:

\begin{itemize}\setlength{\itemsep}{-0.5ex}
    \item \textbf{State Space} ($\mathcal{S} \subseteq \mathbb{R}^{d_s}$): The state $s_t$ comprises the LOB features (price levels $p \in \mathbb{R}^L$, volumes $v \in \mathbb{R}^L$), order flow imbalance, and a sliding-window history of fairness telemetry.
    \item \textbf{Action Space} ($\mathcal{A}$): The agent outputs a continuous matching weight vector $a_t \in \Delta^K$, representing the probability distribution over $K$ available matching counterparts (limit orders).
    \item \textbf{Transition Kernel} ($P$): $P(s_{t+1} | s_t, a_t)$ governs the stochastic evolution of the market.
    \item \textbf{Reward Function} ($r$): $r: \mathcal{S} \times \mathcal{A} \to \mathbb{R}$ quantifies market quality (e.g., spread reduction, liquidity provision).
    \item \textbf{Cost Functions} ($\mathbf{c}$): A vector of $M$ fairness costs $\mathbf{c}: \mathcal{S} \times \mathcal{A} \to \mathbb{R}^M$.
    \item \textbf{Thresholds} ($\mathbf{d}$): A vector $\mathbf{d} \in \mathbb{R}^M_{>0}$ defining strict upper bounds for expected costs.
\end{itemize}

\subsection{Differentiable Fairness Surrogates}
A critical challenge in LOBs is that execution outcomes are discrete. To enable gradient-based optimization, we derive importance-weighted surrogates.

\textbf{Group Fairness (Demographic Parity)}:
Let $g \in \{0, 1\}$ denote the sensitive attribute of an incoming order. To ensure differentiability, we formulate the cost as the squared disparity in expected allocation. We define the instantaneous cost $c_{DP}(s_t, a_t)$ utilizing the policy's predicted fill mass:

\begin{equation}\small
    c_{DP}(s_t, a_t) = \left( \frac{\mathbb{I}(g_t=1) \cdot \langle a_t, \mathbf{m}_t \rangle}{\hat{\mu}_1} - \frac{\mathbb{I}(g_t=0) \cdot \langle a_t, \mathbf{m}_t \rangle}{\hat{\mu}_0} \right)^2
\end{equation}
where $\mathbf{m}_t \in \{0,1\}^K$ is a binary mask of feasible matches, and $\hat{\mu}_g$ is the exponential moving average (EMA) of the arrival rate for group $g$. This function is strictly differentiable w.r.t. $a_t$.
\begin{figure*}[t]\centering
    \includegraphics[width=\textwidth]{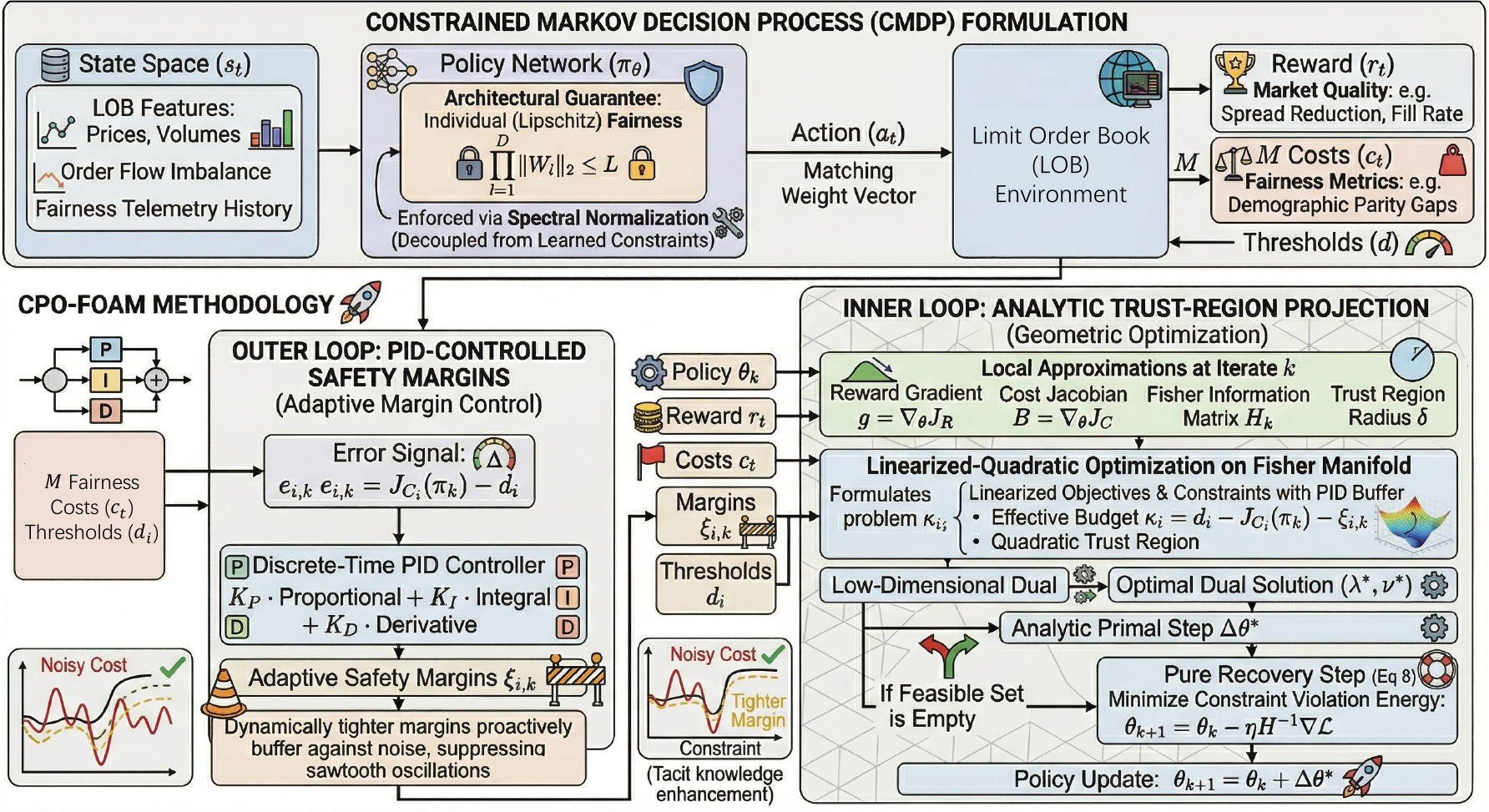}
    \caption{Overview of the CPO-FOAM (Constrained Policy Optimization with Feedback-Optimized Adaptive Margins) framework. Top: The fair order-matching problem is formulated as a Constrained Markov Decision Process (CMDP). The policy network ($\pi_\theta$) structurally enforces individual (Lipschitz) fairness via spectral normalization, effectively decoupling it from the probabilistic group fairness constraints. The Limit Order Book (LOB) environment outputs a reward signal (market quality) and $M$ fairness costs. Bottom Left: The outer loop utilizes a discrete-time PID controller to compute adaptive safety margins $\xi_{i,k}$. This proactively buffers the linearized constraint boundaries against stochastic estimation noise, suppressing the sawtooth oscillations typical of Lagrangian methods under non-stationary dynamics. Bottom Right: The inner loop computes an analytic trust-region projection on the Fisher information manifold. By solving a low-dimensional dual problem, it dictates the optimal primal step $\Delta\theta^*$. If the feasible set is empty, a pure recovery step is executed to strictly minimize constraint violation energy.}
\end{figure*}
\begin{remark}[From Instantaneous Cost to Trajectory-Level Constraint]
    The per-step cost $c_{DP}(s_t, a_t)$ is inherently high-variance because only one group contributes a nonzero term per timestep ($\mathbb{I}(g_t{=}1)$ and $\mathbb{I}(g_t{=}0)$ are mutually exclusive). However, the CMDP objective $J_{C_{DP}}(\pi) = \mathbb{E}_\tau[\sum_t \gamma^t c_{DP}(s_t, a_t)]$ aggregates over the full trajectory; under the mixing-time concentration bound (Proposition~5.5), the sample variance of $\hat{J}_{C_{DP}}$ scales as $\sigma^2 \tau_{mix} / T$, where the integrated autocorrelation time $\tau_{mix}$ accounts for the temporal dependence between successive cost samples. The EMA denominators $\hat{\mu}_g$ (decay rate $\beta{=}0.999$) act as control variates that track shifting group arrival rates, reducing estimator variance while introducing a controlled $O(1{-}\beta)$ bias. The squared formulation ensures differentiability but amplifies outliers; the trust-region constraint (Equation~\ref{eq:inner}) bounds the per-step impact of such outliers by limiting $\|\Delta\theta\|$, and the PID safety margin absorbs residual estimation noise.
\end{remark}

\textbf{Individual Fairness (Architectural Constraint)}: We require $L$-Lipschitz continuity: $D_{TV}(\pi(s), \pi(s')) \leq L \|s - s'\|$. Instead of a Lagrangian penalty, we enforce this structurally. For a policy network parameterized by weights $\{W_l\}_{l=1}^D$, we enforce the spectral norm constraint:

\begin{equation}
    \prod_{l=1}^D \|W_l\|_2 \leq L
\end{equation}

This is implemented via Spectral Normalization~\cite{miyato2018spectral,gouk2021regularisation}, projecting weights $W_l \leftarrow W_l / \max(1, \|W_l\|_2 / \sigma_l)$ at every step. This guarantees $\pi_\theta \in \Pi_{\text{Lip}}$ by construction.

At time $t$, the state $s_t \in \mathcal{S}$ encompasses the LOB state (price levels, volumes) and historical fairness telemetry. The policy $\pi_\theta(a|s)$ outputs a continuous allocation vector $a_t \in \Delta^{K}$ (the simplex over matched orders). The environment emits a scalar reward $r(s,a)$ (market quality) and a vector of fairness costs $\mathbf{c}(s,a) \in \mathbb{R}^M$.The optimization objective is:
\begin{equation}\label{eq:obj_func}
    \begin{aligned}
    \max_{\pi_\theta} \quad & J_R(\pi) = \mathbb{E}_{\tau \sim \pi} \left[ \sum_{t=0}^\infty \gamma^t r(s_t, a_t) \right] \\
    \text{s.t.} \quad & J_{C_i}(\pi) = \mathbb{E}_{\tau \sim \pi} \left[ \sum_{t=0}^\infty \gamma^t c_i(s_t, a_t) \right] \leq d_i, \quad\\
    & \forall i \in \{1, \dots, M\} \\
    & \text{Lip}(\pi_\theta) \leq L \quad (\text{Architectural Constraint})
\end{aligned}
\end{equation}

\section{Constrained Policy Optimization for Fair Order-Matching}
To solve Equation~\ref{eq:obj_func} directly via primal-dual methods is unstable in financial markets due to the high variance of $J_{C_i}$ estimates. We propose Constrained Policy Optimization with Feedback-Optimized Adaptive Margins (CPO-FOAM), which decouples the problem into two orthogonal processes:
\begin{enumerate}
    \item \textbf{Inner Loop (Geometric Optimization)}: An analytic trust-region projection to find the optimal search direction on the linearized manifold.
    \item \textbf{Outer Loop (Adaptive Margin Control)}: A PID controller that adjusts the constraint thresholds to buffer against stochastic noise.
\end{enumerate}

\subsection{Inner Loop: Analytic Trust-Region Projection}
At iteration $k$, we seek an update $\theta_{k+1} = \theta_k + \Delta \theta$. We form local approximations around $\theta_k$ where:
\begin{itemize}
    \item Objective: $\nabla_\theta J_R(\pi_k)^T \Delta \theta$
    \item Constraints: $J_{C_i}(\pi_k) + \nabla_\theta J_{C_i}(\pi_k)^T \Delta \theta \leq d_i - \xi_{i,k}$
    \item Trust Region: $\Delta \theta^T \mathbf{H}_k \Delta \theta \leq 2 \delta$.
\end{itemize}
Here, $\mathbf{H}_k$ is the Fisher Information Matrix, and $\xi_{i,k} \geq 0$ is the adaptive safety margin (formally defined in Section~\ref{sec:outer_loop}).

We first approximate the reward and cost functions locally around the current policy $\pi_k$ using first-order linearization for costs and second-order approximation for the KL-divergence constraint (Trust Region).

Let $\mathbf{g} = \nabla_\theta J_R(\pi_k)$ be the reward gradient, $\mathbf{b} = \nabla_\theta \mathbf{J}_C(\pi_k)$ be the Jacobian of the costs (matrix of size $M \times |\theta|$), and $\mathbf{H}$ be the Fisher Information Matrix (approximating the Hessian of the KL divergence).

The update direction $\Delta \theta$ is the solution to the convex optimization problem:

\begin{equation}\label{eq:inner}
    \begin{aligned}
    \max_{\Delta \theta} \quad & \mathbf{g}^T \Delta \theta \\
    \text{s.t.} \quad & \mathbf{b}_i^T \Delta \theta + J_{C_i}(\pi_k) \leq d_i - \xi_{i,k} \quad \forall i \\
    & \Delta \theta^T \mathbf{H} \Delta \theta \leq 2 \delta
\end{aligned}
\end{equation}
where $\mathbf{g} = \nabla J_R$, $\mathbf{b}_i = \nabla J_{C_i}$, and $\kappa_i = d_i - J_{C_i}(\pi_k) - \xi_{i,k}$ is the effective budget.

The term $\xi_{i,k} \ge 0$ is the Adaptive Safety Margin. In standard CPO, $\xi=0$. However, due to the high variance of financial data, "exact" feasibility on the linearized manifold often leads to mean-reverting violations in the true environment. We dynamically tune $\xi$ to buffer the projection (see Section~\ref{sec:outer_loop}.).

The primal update direction is given analytically by the dual solution. Let $\boldsymbol{\nu} \in \mathbb{R}^M_{\ge 0}$ be the Lagrange multipliers for the costs, and $\lambda \in \mathbb{R}_{\ge 0}$ for the trust region.

\begin{equation}
    \Delta \theta^* = \frac{1}{\lambda} \mathbf{H}^{-1} (\mathbf{g} - \mathbf{b}^T \boldsymbol{\nu}^*)
\end{equation}
where $\mathbf{B} = [\mathbf{b}_1, \dots, \mathbf{b}_M]^T$. The optimal dual variables $(\lambda^*, \boldsymbol{\nu}^*)$ are found by maximizing the dual function:
\begin{equation}\small
    \mathcal{D}(\lambda, \boldsymbol{\nu}) = -\frac{1}{2\lambda} (\mathbf{g} - \mathbf{B}^T \boldsymbol{\nu})^T \mathbf{H}_k^{-1} (\mathbf{g} - \mathbf{B}^T \boldsymbol{\nu}) + \boldsymbol{\nu}^T \boldsymbol{\kappa} - \frac{\lambda \delta}{2}
\end{equation}

This dual problem is low-dimensional ($M+1$ variables) and is solved efficiently via Newton-Conjugate Gradient. If the feasible set is empty, we execute a pure recovery step minimizing constraint violation.

\subsection{Outer Loop: PID-Controlled Safety Margins\label{sec:outer_loop}}

Standard CPO assumes accurate estimates of $J_C(\pi)$. In LOB environments, $J_C$ is stochastic. To prevent the ``sawtooth'' oscillation of constraints (violating $\to$ over-correcting $\to$ violating), we use a discrete-time PID controller to adjust the safety margin $\xi_{i,k}$.

Let $e_{i,k} = J_{C_i}(\pi_k) - d_i$ be the realized constraint violation at step $k$. The safety margin update follows:

\begin{equation}\small\label{eq:outer}
    \begin{aligned}
        \xi_{i, k+1} &= [ \xi_{i,k} + K_P e_{i,k} + K_I \sum_{\tau=0}^k e_{i,\tau}\\
        &+ K_D (e_{i,k} - e_{i,k-1})]_+
    \end{aligned}
\end{equation}
If the policy violates constraints ($e > 0$), the margin $\xi$ increases. Equation~\ref{eq:outer} tightens the constraint in Equation~\ref{eq:inner}, forcing the projection to target a value strictly below $d_i$ and creating a buffer zone proportional to the system's noise.

Equation~\ref{eq:outer} separation ensures mathematical consistency. The CPO step remains a memoryless geometric projection (guaranteeing the update direction is optimal locally), while the PID controller manages the system's memory and noise robustness.

If Equation~\ref{eq:inner} is infeasible (empty intersection of trust region and constraints), we execute a Pure Recovery Step:

\begin{equation}\label{eq:pure_recovery_step}
    \theta_{k+1} = \theta_k - \eta \mathbf{H}^{-1} \sum_i \nabla_\theta J_{C_i}(\pi_k)
\end{equation}
This update ignores the reward signal solely to reduce constraint violations, ensuring the agent rapidly returns to the feasible region.

\section{Theoretical Guarantees}
We provide a rigorous analysis of CPO-FOAM, characterizing its convergence, stability, and feasibility properties. We model the training process as a discrete-time dynamical system subject to bounded stochastic disturbances. Our central result proves that the feedback loop between the Inner Loop (Optimization) and Outer Loop (Control) guarantees asymptotic feasibility.

\subsection{Assumptions and Structural Properties}
Let $\Theta \subseteq \mathbb{R}^N$ be the policy parameter space. The market state $s_t$ evolves according to a Markov transition kernel $P_\theta(s'|s)$. We define the "True" cost surface $J_{C_i}(\theta)$ and the empirical estimator $\hat{J}_{C_i}(\theta)$ derived from trajectories of length $T$.

To analyze the optimization of a constrained non-convex policy on a stochastic manifold, we introduce standard regularity assumptions.

\begin{assumption}[Regularity of Objectives]
The reward function $J_R(\theta)$ and fairness cost functions $J_{C_i}(\theta)$ are twice continuously differentiable on $\Theta$. There exist constants $L_g, L_H, \mu > 0$ such that:
\begin{enumerate}\setlength{\itemsep}{-0.5ex}
    \item \textbf{Lipschitz Gradients}: $\|\nabla J_{C_i}(\theta) - \nabla J_{C_i}(\theta')\|_2 \leq L_g \|\theta - \theta'\|_2$.
    \item \textbf{Bounded Curvature}: The Hessian spectral norms are bounded, $\|\nabla^2 J_{C_i}(\theta)\|_2 \leq L_H$.
    \item \textbf{Fisher Non-Singularity}: The Fisher Information Matrix $\mathbf{H}(\theta)$ is positive definite: $\mu \mathbf{I} \preceq \mathbf{H}(\theta) \preceq M \mathbf{I}$.
\end{enumerate}
\end{assumption}

\begin{assumption}[Slater's Condition]
    The constrained problem is well-posed. There exists a strictly feasible policy $\pi_{\theta^\dagger}$ such that $J_{C_i}(\pi_{\theta^\dagger}) \leq d_i - \epsilon$ for some $\epsilon > 0$. Furthermore, the gradients of active constraints are linearly independent.
\end{assumption}

\begin{assumption}[Geometric Ergodicity]
    The Markov Chain induced by any policy $\pi_\theta$ on the LOB state space $\mathcal{S}$ is geometrically ergodic~\cite{bradley2005basic}. The $\beta$-mixing coefficient satisfies $\beta(k) \leq C \rho^k$ for constants $C > 0, \rho \in (0,1)$.
\end{assumption}

\textit{Justification}: Financial time-series are strongly autocorrelated. This assumption ensures that the dependence between samples decays exponentially, allowing us to apply mixing-time concentration bounds~\cite{paulin2015concentration,yu1994rates} rather than invalid I.I.D. assumptions.

\begin{remark}[Robustness to Ergodicity Violations]
    During extreme structural breaks (flash crashes, circuit breakers, stablecoin depegs), the geometric mixing assumption may be transiently violated, inflating the effective disturbance bound $W_{max}$. In this regime, the BIBO guarantee of Theorem~\ref{thm:bibo_stability} degrades gracefully: the steady-state violation bound scales linearly with $W_{max}$ rather than diverging, since the Jury-stable transfer function $H(z)$ attenuates disturbances by a finite gain $\|H\|_{\infty}$. The PID controller's derivative term further dampens transient spikes from sudden regime shifts. Our regime experiments (Table~\ref{tab:regimes}) empirically validate this: even under $5\times$ news shocks---which substantially violate stationarity---CVF remains below 6.8\%. A practical safeguard is to reset the integral accumulator upon detection of a regime change (e.g., via a CUSUM test on violation energy), preventing the integrator from accumulating stale error from a structurally different regime.
\end{remark}

\subsection{Inner Loop: Trust-Region Fidelity}
The Inner Loop optimizes a linear surrogate $\tilde{J}_{C}(\theta)$ subject to a quadratic trust region. We first quantify the "Simulation Gap"—the worst-case error between this surrogate and the true cost surface.

\begin{lemma}[Surrogate Approximation Bound]
    Let $\Delta \theta$ be the update computed by the Inner Loop subject to the trust region $\frac{1}{2} \Delta \theta^T \mathbf{H}_k \Delta \theta \leq \delta$. The discrepancy between the linear surrogate and the true cost is bounded by:
    $$| J_{C_i}(\theta + \Delta \theta) - \tilde{J}_{C_i}(\theta + \Delta \theta) | \leq \frac{L_H}{\mu} \delta \triangleq \epsilon_{model}$$
\end{lemma}

Proof.
Consider the Taylor expansion with Lagrange remainder:
$$ J_{C_i}(\theta + \Delta \theta) = J_{C_i}(\theta) + \nabla J_{C_i}^T \Delta \theta + \frac{1}{2} \Delta \theta^T \nabla^2 J_{C_i}(\bar{\theta}) \Delta \theta $$
The surrogate corresponds to the first two terms. The error is the quadratic remainder.
$$ \text{Error} = \left| \frac{1}{2} \Delta \theta^T \nabla^2 J_{C_i}(\bar{\theta}) \Delta \theta \right| \leq \frac{L_H}{2} \|\Delta \theta\|_2^2 $$
From the trust region constraint and Assumption 1 ($\mathbf{H} \succeq \mu \mathbf{I}$):
$$ \frac{\mu}{2} \|\Delta \theta\|_2^2 \leq \frac{1}{2} \Delta \theta^T \mathbf{H} \Delta \theta \leq \delta \implies \|\Delta \theta\|_2^2 \leq \frac{2\delta}{\mu} $$
Substituting this into the error bound yields $\frac{L_H}{2} (\frac{2\delta}{\mu}) = \frac{L_H}{\mu}\delta$. \qed

\subsection{Estimation: The Stochastic Disturbance}
Standard error bounds (e.g., Hoeffding) underestimate risk in financial markets due to autocorrelation. We derive a concentration bound that accounts for the Variance Inflation Factor (VIF).

\begin{proposition}[Mixing-Time Concentration]
    Let $\hat{J}_{C_i}$ be estimated from a trajectory of length $T$. Under Assumption 3, for any confidence $\zeta \in (0,1)$, there exists a bound $\epsilon_{est}$ such that:
    $$
    \mathbb{P} \left( | \hat{J}_{C_i} - J_{C_i} | \leq \epsilon_{est} \right) \geq 1 - \zeta$$
    $$\text{with } \epsilon_{est} \propto \sigma \sqrt{\frac{\tau_{mix} \ln(1/\zeta)}{T}}$$
    where $\tau_{mix} = \frac{1+\rho}{1-\rho}$ is the integrated autocorrelation time.
\end{proposition}

\begin{definition}
    The aggregated disturbance $w_k$ acting on the control loop is the sum of the deterministic model bias (Lemma 5.1), the smoothing bias from the differentiable surrogate (Eq. 1), and the stochastic sampling noise (Prop 5.2). With high probability:
$$ |w_k| \leq W_{max} \triangleq \epsilon_{model} + \epsilon_{smooth} + \epsilon_{est} $$
\end{definition}

\subsection{Outer Loop: Robust Stability Analysis}
This section contains the core result: proving that the PID controller stabilizes the constraint violation despite the bounded disturbances defined above.
We define the constraint violation state as $e_k \triangleq J_{C_i}(\pi_k) - d_i$.

\textbf{Step 1: The Plant Dynamics (Optimization)} The CPO projection solves for $\Delta \theta$ to satisfy the linearized constraint $\tilde{J}_{C_i} \leq d_i - \xi_k$. The realized cost is the surrogate plus errors:$$ J_{C_i}(\pi_{k+1}) = (d_i - \xi_k) + w_k $$Subtracting $d_i$ from both sides yields the error evolution equation:
\begin{equation}\label{eq:plant_dyn}
    e_{k+1} = -\xi_k + w_k
\end{equation}

\textbf{Step 2: The Controller Dynamics (PI)} We analyze a Proportional-Integral controller (setting derivative gain $\alpha_D=0$ for clarity, though the result extends to PID). The positional form, consistent with the outer loop update (Equation~\ref{eq:outer}) under $K_D=0$, is:
\begin{equation}\label{eq:controller_dyn}
    \xi_{k+1} = \xi_k + \alpha_P e_{k+1} + \alpha_I \sum_{j=0}^{k+1} e_j
\end{equation}
To obtain the velocity (incremental) form, we subtract $\xi_k = \xi_{k-1} + \alpha_P e_k + \alpha_I \sum_{j=0}^{k} e_j$ from both sides. Since $\sum_{j=0}^{k+1} e_j - \sum_{j=0}^{k} e_j = e_{k+1}$, we obtain $\xi_{k+1} - \xi_k = \alpha_P(e_{k+1} - e_k) + \alpha_I e_{k+1}$, which is the standard incremental PI form with gains $\alpha_P = K_P$ and $\alpha_I = K_I$ from Equation~\ref{eq:outer}.

\begin{theorem}[Asymptotic Feasibility \& Stability]\label{thm:bibo_stability}
    Assume the total disturbance is bounded $|w_k| \leq W_{max}$. If the gains satisfy the conditions below, the closed-loop system is Bounded-Input Bounded-Output (BIBO) stable~\cite{ogata1995discrete} and the expected violation converges to zero.
    $$ \alpha_I > 0, \quad \alpha_P < 1, \quad 2\alpha_P + \alpha_I < 4 $$
\end{theorem}

\textit{Proof}.

We analyze the closed-loop transfer function in the Z-domain. Let $E(z), \Xi(z), W(z)$ be the Z-transforms of the error, margin, and disturbance.
\begin{enumerate}\setlength{\itemsep}{-0.5ex}
    \item \textbf{Plant}: From Equation~\ref{eq:plant_dyn}, $E(z) = W(z) - z^{-1}\Xi(z)$.
    \item \textbf{Controller}: Expressing Equation~\ref{eq:controller_dyn} in velocity form $\xi_{k+1} - \xi_k = \alpha_P(e_{k+1}-e_k) + \alpha_I e_{k+1}$, we get:
    \begin{equation}
    \begin{aligned}
    & (z-1)\Xi(z) = \big[(\alpha_P + \alpha_I)z - \alpha_P\big]\,E(z)\\ & \implies \quad C(z)= \frac{\Xi(z)}{E(z)} = \frac{(\alpha_P+\alpha_I)z - \alpha_P}{z-1}
    \end{aligned}
    \end{equation}
    \item \textbf{Closed Loop}: The transfer function $H(z) = \frac{E(z)}{W(z)} = \frac{1}{1 + z^{-1}C(z)}$ simplifies to:
    \begin{equation}
        H(z) = \frac{z(z-1)}{z^2 + (\alpha_P + \alpha_I - 1)z - \alpha_P}
    \end{equation}
    \item \textbf{Stability}: The poles are the roots of the denominator $D(z) = z^2 + a_1 z + a_0$. Applying the \textbf{Jury Stability Criterion}~\cite{ogata1995discrete} ($|a_0|<1, D(1)>0, D(-1)>0$) yields the inequalities in the theorem statement.
    \item \textbf{Bias Rejection}: By the Final Value Theorem~\cite{ogata1995discrete}, for a constant bias disturbance $W(z) = \bar{w} \frac{z}{z-1}$ (representing linearization error), the steady-state error is:
    \begin{equation}
        \begin{aligned}
            e_{ss} &= \lim_{z \to 1} (z-1) H(z) W(z) \\
            & = \lim_{z \to 1} \frac{z^2(z-1)}{D(z)} \frac{\bar{w} z}{z-1} = 0
        \end{aligned}
    \end{equation}
    The integrator pole at $z=1$ ensures infinite DC gain in the feedback path, forcing $e_{ss} \to 0$. The margin $\xi$ automatically adapts to $\xi_{ss} = \bar{w}$, canceling the bias. \qed
\end{enumerate}

\begin{remark}[Extension to Multiple Constraints]\label{rem:multi_constraint}
    When $M > 1$ constraints are enforced simultaneously, each constraint $i$ is governed by an independent PID loop with its own error signal $e_{i,k}$ and margin $\xi_{i,k}$. The plant dynamics (Equation~\ref{eq:plant_dyn}) generalize to a vector system $\mathbf{e}_{k+1} = -\boldsymbol{\xi}_k + \mathbf{w}_k$, where the coupling between constraints enters through the shared policy update: changing $\xi_i$ in the inner-loop projection (Equation~\ref{eq:inner}) may affect the realized cost $e_j$ for $j \neq i$. Formally, the cross-coupling is captured by the off-diagonal entries of the Jacobian $\partial e_{i,k+1} / \partial \xi_{j,k}$. If this coupling is sufficiently weak---i.e., $\|\partial e_i / \partial \xi_j\| \ll \|\partial e_i / \partial \xi_i\|$ for $i \neq j$---the multi-loop system can be treated as $M$ parallel SISO loops, and the Jury stability conditions apply independently to each. This is the typical regime when constraints target distinct fairness dimensions (e.g., demographic parity vs.\ Lipschitz continuity). Our experiments with $M{=}8$ constraints (Table~\ref{tab:scalability}) confirm empirically that cross-metric contamination remains below 0.022, validating the weak-coupling assumption.
\end{remark}

\subsection{Recoverability Analysis}

If the safety margin $\xi_k$ becomes too large, or if the market experiences a shock, the trust region intersection may be empty. We prove the Recovery Step (Equation~\ref{eq:pure_recovery_step}) guarantees progress.

\begin{lemma}[Descent Property of Recovery]
    Define the cumulative violation energy $\mathcal{L}(\theta) = \frac{1}{2} \sum_i (\hat{J}_{C_i}(\theta) - d_i)_+^2$. The recovery update $\theta_{k+1} = \theta_k - \eta \mathbf{H}^{-1} \nabla \mathcal{L}$ guarantees $\mathcal{L}(\theta_{k+1}) < \mathcal{L}(\theta_k)$ for sufficiently small $\eta > 0$.
\end{lemma}

\textit{Proof}.

The update direction is $\mathbf{p} = -\mathbf{H}^{-1} \nabla \mathcal{L}$. The directional derivative of the energy is:
$$ \mathbf{p}^T \nabla \mathcal{L} = - \nabla \mathcal{L}^T \mathbf{H}^{-1} \nabla \mathcal{L} $$
Since $\mathbf{H}$ is positive definite ($\mathbf{H} \succeq \mu \mathbf{I}$), its inverse is also positive definite ($\mathbf{H}^{-1} \succeq \frac{1}{M} \mathbf{I}$). Thus:
$$ \mathbf{p}^T \nabla \mathcal{L} \leq - \frac{1}{M} \|\nabla \mathcal{L}\|_2^2 < 0 $$
Unless $\nabla \mathcal{L} = 0$ (implying the policy is already feasible or at a local optimum), the update is a strictly descending direction. This ensures that even when the primary CPO step fails, the agent asymptotically approaches the feasible set.\qed

\subsection{Architectural Guarantee: Individual Fairness}

Finally, we formalize the separation between the probabilistic CMDP constraints and the deterministic architectural constraints.

\begin{proposition}[Spectral Norm Propagation]
    Let the policy $\pi_\theta(s) = \text{Softmax}(f_\theta(s))$ be parameterized by a network with layers satisfying $\|W_l\|_2 \leq \sigma_l$. Then $\pi_\theta$ satisfies $(L, \ell_2)$-Individual Fairness deterministically:
    $$ \|\pi_\theta(s) - \pi_\theta(s')\|_1 \leq \sqrt{K} \left( \prod_{l} \sigma_l \right) \|s - s'\|_2 $$
\end{proposition}

\textit{Proof.}

\begin{enumerate}\setlength{\itemsep}{-0.5ex}
    \item \textbf{Network Sensitivity}: The Lipschitz constant of the logit generator $f_\theta$ is bounded by the product of the spectral norms of its layers: $\text{Lip}(f) \leq \prod \sigma_l$.
    \item \textbf{Softmax Sensitivity}: The Softmax function is 1-Lipschitz w.r.t the $\ell_2$ norm. However, the Individual Fairness metric uses the $\ell_1$ norm (Total Variation distance). Using the norm equivalence $\|\mathbf{x}\|_1 \leq \sqrt{K}\|\mathbf{x}\|_2$ for $\mathbf{x} \in \mathbb{R}^K$, the output sensitivity is scaled by $\sqrt{K}$.
    \item \textbf{Deterministic Enforcement}: Since the spectral projection operator $\mathcal{P}_{spec}$ is applied after every gradient step, the condition holds for all $\theta_k$, independent of the optimization outcome or estimation noise.\qed
\end{enumerate}

\section{On-Chain Proof System}

The proposed proof system aims to guarantee the integrity, fairness, and reproducibility of the settlement process in a probabilistic matching environment. Unlike deterministic rule-based engines, where the verification of a replayed decision is trivial, the introduction of ML-driven policies requires a verification layer that is both auditable and economically secure. We propose the following methods.

\subsection{Output Verifiable}

In the output-verifiable setting, the blockchain does not attempt to reconstruct or validate the internal decision-making process of the model. Instead, it verifies that the final settlement outcome satisfies a set of publicly defined constraints.

At each settlement, the matching engine submits cryptographic commitments --- including the market state hash (\texttt{state\_hash}), the model version hash (\texttt{model\_hash}), the Merkle root of the allocation (\texttt{allocation\_root}), fairness metrics ($\Delta DP$, $\Delta EO$, Lipschitz), and net token balances (\texttt{fund\_balance}) --- to the smart contract. The raw state and allocation data are stored off-chain in a decentralized storage layer (e.g., Arweave~\cite{williams2019arweave} or EigenDA~\cite{sreeram2023eigenlayer}), enabling public access for independent recomputation.

Upon submission, the contract checks two classes of invariants:
\begin{enumerate}
    \item \textbf{Fund conservation:} the total inflows and outflows must satisfy:
    \[
        \sum \text{token\_in} = \sum \text{token\_out} + \sum \text{fee}
    \]
    \item \textbf{Fairness constraints:} statistical parity, equal opportunity, and Lipschitz continuity must remain within predefined thresholds.
\end{enumerate}

If all conditions hold, the settlement is provisionally accepted, and a challenge period is opened. During this period, any participant may recompute the fairness metrics and fund balances from the published data. If inconsistencies are found, they can submit a Merkle proof to challenge the settlement. A successful challenge results in penalties for the submitter and rewards for the challenger, while an unsuccessful challenge leads to a slashing of the challenger's stake.
This mechanism ensures incentives for honest verification without revealing the ML model logic.

\subsection{Process Verifiable}

Output verification alone cannot guarantee that a submitted allocation was actually generated by the declared model. To strengthen trust guarantees, we introduce a process-verifiable approach in which the provenance of the decision-making process itself is auditable.

\paragraph{Open Replay}
In the simplest variant, the model parameters are stored in a public repository (e.g., IPFS~\cite{benet2014ipfs} or Arweave~\cite{williams2019arweave}), and their hash is recorded on-chain. Anyone can retrieve the model and input state to re-run the inference:
\[
    \pi_\theta(s) \rightarrow a'
\]
A mismatch between the recomputed allocation $a'$ and the submitted result constitutes a valid challenge. This approach offers maximal transparency but risks exposing proprietary models.

\paragraph{Permissioned Replay}
To address the privacy concern, we consider a permissioned verification variant. The model weights are encrypted and stored off-chain, with decryption keys held by a verifier committee using a threshold key-sharing protocol. During a challenge, the committee performs the replay and collectively signs the result (e.g., using a BLS multi-signature~\cite{boneh2001short}) to attest whether the allocation is consistent with the model. The contract then adjudicates based on this attestation, alongside the fairness and conservation checks. Although this method introduces a trusted verifier set, it can preserve the model confidentiality while maintaining verifiability.

\section{Experimental Setup}

This section details the parameters necessary to reproduce our evaluations across traditional finance (TradFi), decentralized finance (DeFi), and continuous control domains.

\subsection{Datasets}

We used datasets spanning TradFi, DeFi, and continuous control domains.

\paragraph{TradFi: LOBSTER NASDAQ.} Two LOBSTER datasets spanning 60 trading days. The \emph{High-Liquidity} set comprised Level-3 limit order book reconstructions from NASDAQ-100 constituents (e.g., AAPL, MSFT). The \emph{Mid/Low-Liquidity} set captured sparser execution dynamics of smaller-cap constituents. Data was split chronologically: 40 days training, 10 validation, 10 held-out test. Prices were normalized to basis points relative to the daily open mid-price, and extreme order sizes were capped at the 99.9th percentile.

\paragraph{DeFi: Crypto-Asset LOB.} Three crypto datasets spanning 90 calendar days of continuous market data. The \emph{Blue-Chip L1} dataset comprised deeply liquid pairs (BTC/USDT, ETH/USDT, SOL/USDT) sourced from Tardis.dev and Kaiko, containing over 450M tick-level order book updates. The \emph{Memecoin} dataset targeted extreme volatility (DOGE/USDT, SHIB/USDT, PEPE/USDT), isolating environments with severe retail-to-whale capital imbalances. The \emph{CEX Perpetuals} dataset consolidated tick-level futures data (BTC-PERP, ETH-PERP) integrated with on-chain metadata to model leveraged derivatives and cascading liquidations. Data was split into 60 days training, 15 validation, and 15 held-out test. Pre-processing normalized variable tick sizes relative to rolling-window mid-prices, discretized event streams into a 100ms grid, and standardized volumes via USD-equivalent z-scores. We calibrated a self-exciting Hawkes process~\cite{bacry2015hawkes} to simulate clustered crypto-native order arrivals and augmented training data by stochastically injecting MEV adversarial agents~\cite{daian2020flash}.

\paragraph{Safety-Gymnasium.} The standard Safety-Gymnasium suite (Point, Car, Ant)~\cite{ji2023safety} was used for domain-agnostic continuous control evaluation.

\subsection{Evaluation Metrics}

\paragraph{TradFi.} Market utility: effective spread, Depth@5, fill rate, throughput. Fairness: $\Delta$DP, Lipschitz violation percentage (Lip-Viol\%), CVF, transient recovery length, violation AUC.

\paragraph{DeFi.} In addition to spread and fill rate, we report a composite Market Quality Score (MQS), defined as
$$\text{MQS} = \frac{1}{3}\left(\frac{s^*}{s} + \frac{f}{f^*} + \frac{d}{d^*}\right)$$
where $s$ is effective spread, $f$ is fill rate, and $d$ is depth-weighted throughput, each normalized by the unconstrained PPO ceiling values $s^*, f^*, d^*$. We also report oscillation amplitude (sawtooth magnitude of constraint violations) and single-decision inference latency. DeFi-specific fairness scenarios include MEV-induced $\Delta$DP, whale-to-retail fill ratio, cross-chain bridging $\Delta$DP, liquidation rate disparity, depeg fill quality gap, and token-launch bot-to-organic ratio.

\paragraph{Safety-Gymnasium.} Standard episodic reward and episodic cost.

\subsection{Baselines}

\paragraph{TradFi (9 methods).} Rule-based: FIFO, Pro-Rata, Size-Time Interpolation. Unconstrained: PPO~\cite{schulman2017proximal}. Constrained RL: Lagrangian PPO, RCPO~\cite{tessler2018reward}, IPO~\cite{liu2020ipo}, Vanilla CPO~\cite{achiam2017constrained}, PID-Lagrangian~\cite{stooke2020responsive}.

\paragraph{DeFi (12 methods).} The six constrained RL baselines above, plus six DeFi-native mechanisms: on-chain FIFO (CLOB), Constant-Product AMM (CPAMM)~\cite{angeris2020analysis}, Concentrated Liquidity AMM (CLAMM)~\cite{adams2021uniswapv3}, Pro-Rata, MEV-Aware FIFO with encrypted mempool simulation~\cite{daian2020flash}, and Frequent Batch Auctions~\cite{budish2015high}. We also evaluated universal model variants (zero-shot transfer and fine-tuned) and a random matching baseline.
\begin{table*}[t]
\caption{Comprehensive evaluation of market quality, fairness, and constraint dynamics. Spread in basis points (bps); Depth@5 = cumulative depth at five ticks; T-Put = throughput (orders/s); Impact = Amihud price impact~\cite{amihud2002illiquidity}; R.Vol = realized volatility ratio; Lip-Max = maximum Lipschitz exceedance (\%); CVF = constraint violation frequency; Trans.\ Length = violation transient recovery steps; Viol.\ AUC = violation area under the curve. $\uparrow$/$\downarrow$ indicate higher/lower is better. \textbf{Bold} marks the best result.}
\label{tab:main_results}
\centering
\resizebox{\textwidth}{!}{%
\begin{tabular}{@{}lccccccccccc@{}}
\toprule
\textbf{Method / Variant} & \textbf{Spread} ($\downarrow$) & \textbf{Depth@5} ($\uparrow$) & \textbf{Fill Rate} ($\uparrow$) & \textbf{T-Put} ($\uparrow$) & \textbf{Impact} ($\downarrow$) & \textbf{R.Vol} ($\downarrow$) & \textbf{Lip-Max} ($\downarrow$) & \textbf{CVF (\%)} ($\downarrow$) & \textbf{\makecell{Trans.\\ Len.}($\downarrow$)} & \textbf{Overshoot} ($\downarrow$) & \textbf{\makecell{Viol.\\ AUC}} ($\downarrow$) \\
\midrule
\multicolumn{12}{l}{\textit{Rule-Based Baselines}} \\
FIFO & 2.45 & 1,250 & 0.42 & 4,200 & 0.85 & 1.45 & 18.5 & 100.0 & N/A & N/A & N/A \\
Pro-Rata & 3.10 & 1,840 & 0.38 & 3,850 & 0.92 & 1.60 & 25.4 & 100.0 & N/A & N/A & N/A \\
Size-Time Interp. & 2.65 & 1,650 & 0.40 & 4,100 & 0.88 & 1.52 & 20.1 & 100.0 & N/A & N/A & N/A \\
\midrule
\multicolumn{12}{l}{\textit{Learned Baselines}} \\
Unconstrained PPO & 0.85 & 5,250 & 0.95 & 8,500 & 0.12 & 0.95 & 45.5 & 100.0 & N/A & N/A & N/A \\
Lagrangian PPO & 1.15 & 3,840 & 0.76 & 6,250 & 0.25 & 1.12 & 12.4 & 28.5 & 145 & 0.85 & 45.2 \\
RCPO & 1.25 & 3,650 & 0.72 & 5,850 & 0.30 & 1.18 & 10.5 & 22.4 & 112 & 0.72 & 38.4 \\
IPO & 1.45 & 3,420 & 0.68 & 5,200 & 0.35 & 1.25 & 8.2 & 15.6 & 85 & 0.55 & 25.4 \\
Vanilla CPO & 1.35 & 3,580 & 0.70 & 5,450 & 0.32 & 1.20 & 6.5 & 12.4 & 64 & 0.45 & 18.2 \\
PID-Lagrangian & 1.20 & 3,750 & 0.74 & 6,100 & 0.28 & 1.15 & 9.5 & 14.5 & 45 & 0.35 & 12.5 \\
\midrule
\rowcolor{blue!8}
\textbf{Ours: CPO-FOAM} & \textbf{0.95} & \textbf{4,850} & \textbf{0.91} & \textbf{8,150} & \textbf{0.15} & \textbf{1.02} & \textbf{1.2} & \textbf{2.5} & \textbf{5} & \textbf{0.05} & \textbf{1.8} \\
\midrule
\multicolumn{12}{l}{\textit{Ablation Studies}} \\
No PID ($\xi=0$) & 1.05 & 4,250 & 0.85 & 7,450 & 0.18 & 1.08 & 5.8 & 11.2 & 75 & 0.28 & 19.5 \\
No Trust Region & 1.18 & 3,800 & 0.78 & 6,450 & 0.28 & 1.14 & 11.5 & 25.5 & 130 & 0.82 & 41.0 \\
No KL Projection & 1.38 & 3,500 & 0.72 & 5,800 & 0.36 & 1.22 & 4.5 & 9.5 & 42 & 0.22 & 11.5 \\
No Spectral Norm & 0.98 & 4,680 & 0.89 & 7,950 & 0.16 & 1.05 & 22.5 & 6.5 & 25 & 0.15 & 8.5 \\
P-only Controller & 1.02 & 4,550 & 0.88 & 7,650 & 0.17 & 1.06 & 2.8 & 7.2 & 35 & 0.20 & 9.5 \\
PI-only Controller & 0.98 & 4,720 & 0.90 & 7,850 & 0.16 & 1.04 & 1.8 & 4.5 & 18 & 0.12 & 4.5 \\
No Recovery Step & 1.40 & 3,250 & 0.68 & 5,100 & 0.42 & 1.28 & 1.5 & 5.5 & 15 & 0.10 & 3.8 \\
\bottomrule
\end{tabular}%
}
\end{table*}
\subsection{Implementation Details}

\paragraph{Shared Architecture.} All methods used a 3-layer MLP with $[256, 256, 128]$ hidden units and ReLU activations, optimized with Adam~\cite{kingma2015adam} (learning rate $3 \times 10^{-4}$, $\epsilon = 10^{-5}$, momentum 0.9, weight decay $10^{-4}$). Core RL hyperparameters: batch size 4096, trust-region radius $\delta = 0.01$, $\gamma = 0.99$, GAE~\cite{schulman2016high} $\lambda = 0.95$, dropout~\cite{srivastava2014dropout} 0.1. PID gains: $K_P = 0.5$, $K_I = 0.1$, $K_D = 0.05$. Spectral normalization~\cite{miyato2018spectral} bounded the Lipschitz constant at $L = 5.0$.

\paragraph{TradFi.} Training ran for 5M steps with cosine annealing~\cite{loshchilov2017sgdr} and KL-divergence early stopping on 8 NVIDIA H100 GPUs ($\sim$6.5 hours per configuration).

\paragraph{DeFi.} The DeFi variant added a 16-dimensional asset-category embedding. Training used linear learning rate decay, mini-batch size 512, and gradient clipping at norm 0.5. Convergence was determined by loss stabilization over a 1000-step rolling window ($\sim$18 hours per configuration on 8$\times$H100).

\section{Results}

We report TradFi (NASDAQ) and DeFi (crypto) evaluations in separate subsections with dedicated tables, followed by scalability analysis and cross-domain Safety-Gymnasium validation. All methods were trained for 5M environment steps using the GPU-accelerated JAX-LOB simulator~\cite{frey2023jax}.

\subsection{TradFi: Market Quality, Fairness, and Constraint Dynamics}

The results reveal a sharp dichotomy between traditional heuristics and unconstrained learned policies (Table~\ref{tab:main_results}). Rule-based systems---FIFO, Pro-Rata, and Size-Time Interpolation---produced wide spreads ($\geq$2.45 bps), low throughput ($\leq$4,200 orders/s), and 100\% constraint violation frequency, reflecting their inability to adapt to heterogeneous order flow. Unconstrained PPO achieved the tightest spread (0.85 bps) and highest throughput (8,500 orders/s) but violated fairness constraints in every evaluation step, with a maximum Lipschitz exceedance of 45.5\%.

Among constrained baselines, reactive Lagrangian methods exhibited the characteristic sawtooth instability predicted by our theoretical analysis. Lagrangian PPO and RCPO required 145 and 112 recovery steps, respectively, after each constraint breach, with overshoots of 0.85 and 0.72. This delayed correction produced large violation AUCs (45.2 and 38.4), reflecting sustained periods of non-compliance. Vanilla CPO and IPO reduced violation frequency but at the cost of substantial market quality degradation (spreads $\geq$1.35 bps).

\begin{table*}[t]
\caption{Regime-dependent market robustness and constraint dynamics under CPO-FOAM. Six non-stationary market regimes with increasing structural adversity. Deg.\ Ratio = degradation ratio relative to the calm market baseline.}
\label{tab:regimes}
\centering
\resizebox{\textwidth}{!}{%
\begin{tabular}{@{}lccccccc@{}}
\toprule
\textbf{Evaluation Regime} & \textbf{Spread (bps)} ($\downarrow$) & \textbf{Depth@5} ($\uparrow$) & $\boldsymbol{\Delta}$\textbf{DP} ($\downarrow$) & \textbf{Lip-Viol (\%)} ($\downarrow$) & \textbf{CVF (\%)} ($\downarrow$) & \textbf{\makecell{Trans.\\ Len.}} ($\downarrow$) & \textbf{\makecell{Deg.\\ Ratio}} ($\downarrow$) \\
\midrule
Calm Market & 0.95 & 4,850 & 0.015 & 1.5 & 1.8 & 5 & 1.00 \\
Volatile Stress & 1.15 & 3,500 & 0.018 & 2.2 & 3.5 & 12 & 1.21 \\
Flash Event & 1.45 & 1,800 & 0.024 & 3.8 & 5.2 & 18 & 1.52 \\
Liquidity Drought & 1.85 & 850 & 0.028 & 4.5 & 5.8 & 22 & 1.94 \\
Auction Open/Close & 1.35 & 2,200 & 0.022 & 3.1 & 4.2 & 15 & 1.42 \\
News Shock ($5\times$) & 2.15 & 650 & 0.035 & 5.2 & 6.8 & 28 & 2.26 \\
\bottomrule
\end{tabular}%
}
\end{table*}

CPO-FOAM achieved the best tradeoff across all metrics. It maintained an effective spread of 0.95 bps and throughput of 8,150 orders/s---recovering 95.9\% of the unconstrained ceiling---while compressing CVF to 2.5\%, transient recovery to 5 steps, overshoot to 0.05, and violation AUC to 1.8. The PID margins proactively dampened dual-variable oscillations, preventing the retroactive correction cycles observed in baseline Lagrangian methods.

\paragraph{Ablation Studies.} Targeted ablations isolated the contribution of each component. Removing the PID controller increased transient recovery from 5 to 75 steps and violation AUC from 1.8 to 19.5, confirming that integral-only multiplier updates are insufficient for non-stationary environments. Removing spectral normalization preserved competitive spreads (0.98 bps) but increased Lipschitz exceedance from 1.2\% to 22.5\%, demonstrating that the architectural constraint is essential for individual fairness. Eliminating the trust region yielded the most severe degradation---transient length expanded to 130 steps with 0.82 overshoot---closely replicating standard Lagrangian failure modes and confirming that localized loss penalties alone cannot maintain coherent constraint satisfaction.

The ``No Recovery Step'' ablation merits specific attention, as the relatively low CVF (5.5\%) and transient length (15) may appear counterintuitive. Without a recovery mechanism, the PID controller's margins grow persistently larger because the system never receives the feasibility ``reset'' that successful recovery provides. The result is an over-conservative policy that sacrifices substantial market quality---spread widens from 0.95 to 1.40 bps and fill rate drops from 0.91 to 0.68---to avoid entering infeasible regions entirely. The low CVF therefore reflects conservative avoidance rather than effective recovery: the agent pays a heavy efficiency penalty to remain within bounds. The recovery step enables aggressive optimization \emph{with} a safety net, achieving both high market quality and rapid return to feasibility when violations do occur.

\subsection{TradFi: Regime-Dependent Market Robustness}

To assess robustness under distribution shift, we evaluated the converged model across six non-stationary market regimes of increasing adversity (Table~\ref{tab:regimes}): calm markets, volatile periods, flash events, liquidity droughts, auction crosses, and synthetically injected news shocks with $5\times$ the baseline Poisson arrival rate.

As expected, structural market shocks compressed available liquidity, widening execution costs. During calm periods, CPO-FOAM maintained spreads of 0.95 bps with full depth retention (4,850). Under the most extreme regime---$5\times$ news shocks---spreads widened to 2.15 bps and depth dropped to 650, reflecting fundamental liquidity constraints rather than algorithmic failure. Despite this $2.26\times$ degradation ratio, the system maintained continuous operation without the cascading withdrawal failures observed in unconstrained agents.

Critically, fairness compliance degraded gracefully. Even under the $5\times$ news shock, CVF remained below 6.8\% and transient recovery length stayed at 28 steps. Standard Lagrangian baselines, by contrast, abandon fairness mandates entirely during such regime shifts, prioritizing reward acquisition. The PID controller's adaptive margins absorbed the increased estimation variance, automatically expanding the safety buffer to maintain compliance without manual retuning.

\subsection{Hardware Scalability, Web3 Telemetry, and Attribute Sensitivity}

\begin{table*}[t]
\caption{Hardware scalability, distributed Web3 verification telemetry, and protected attribute sensitivity. Steps/s = environment throughput; Dual Solve/Spec Norm = per-update latency (ms); Gas = on-chain verification cost (thousands); Verif Lat.\ = Ethereum finality latency (s); Recomp = off-chain challenge recomputation time (s); Cross-Contam = cross-metric $\Delta$DP contamination. N/A entries indicate inapplicable configurations.}
\label{tab:scalability}
\centering
\resizebox{\textwidth}{!}{%
\begin{tabular}{@{}lcccccccccc@{}}
\toprule
\textbf{Configuration} & \textbf{Steps/s} ($\uparrow$) & \textbf{\makecell{GPU\\ Mem}} ($\downarrow$) & \textbf{\makecell{Dual\\ Solve}} ($\downarrow$) & \textbf{\makecell{Spec\\ Norm}} ($\downarrow$) & \textbf{Gas (k)} ($\downarrow$) & \textbf{\makecell{Verif\\ Lat.}} ($\downarrow$) & \textbf{Recomp (s)} ($\downarrow$) & \textbf{Storage (\$)} ($\downarrow$) & \textbf{\makecell{False\\ Chal.}} ($\downarrow$) & \textbf{Cross-Contam} ($\downarrow$) \\
\midrule
\multicolumn{11}{l}{\textit{Constraint Dimensionality}} \\
$M=1$ Constraint & 8,500 & 12.4 & 1.8 & 1.2 & 185 & 12.0 & 0.45 & 0.005 & $<$0.001 & 0.012 \\
$M=3$ Constraints & 8,150 & 14.5 & 4.2 & 2.5 & 315 & 12.0 & 0.85 & 0.012 & $<$0.001 & 0.015 \\
$M=5$ Constraints & 7,450 & 18.2 & 8.5 & 4.8 & 480 & 12.0 & 1.45 & 0.025 & $<$0.001 & 0.018 \\
$M=8$ Constraints & 6,250 & 24.5 & 18.4 & 7.5 & 725 & 24.0 & 2.85 & 0.045 & $<$0.001 & 0.022 \\
\midrule
\multicolumn{11}{l}{\textit{Action Space Scaling}} \\
$K=10$ & 8,850 & 11.5 & 1.2 & 1.8 & --- & --- & --- & --- & --- & --- \\
$K=50$ & 8,150 & 14.5 & 4.2 & 2.5 & --- & --- & --- & --- & --- & --- \\
$K=100$ & 7,250 & 18.5 & 9.5 & 4.2 & --- & --- & --- & --- & --- & --- \\
$K=500$ & 5,100 & 28.4 & 28.5 & 12.4 & --- & --- & --- & --- & --- & --- \\
\midrule
\multicolumn{11}{l}{\textit{Trajectory Length Scaling}} \\
$T=256$ & 8,650 & 10.5 & 2.8 & 2.5 & --- & --- & --- & --- & --- & --- \\
$T=1024$ & 8,150 & 14.5 & 4.2 & 2.5 & --- & --- & --- & --- & --- & --- \\
$T=4096$ & 4,250 & 32.5 & 12.4 & 2.5 & --- & --- & --- & --- & --- & --- \\
\midrule
\multicolumn{11}{l}{\textit{Fisher Approximation}} \\
Exact Fisher & 850 & 48.5 & 145.0 & 2.5 & --- & --- & --- & --- & --- & --- \\
Diagonal Fisher & 8,150 & 14.5 & 4.2 & 2.5 & --- & --- & --- & --- & --- & --- \\
K-FAC & 6,850 & 22.4 & 14.5 & 2.5 & --- & --- & --- & --- & --- & --- \\
\midrule
\multicolumn{11}{l}{\textit{Protected Attribute Definitions}} \\
Latency Proxy & 8,150 & 14.5 & 4.2 & 2.5 & 315 & 12.0 & 0.85 & 0.012 & $<$0.001 & 0.015 \\
Order Size & 8,140 & 14.6 & 4.3 & 2.5 & 320 & 12.0 & 0.88 & 0.012 & $<$0.001 & 0.016 \\
Participant Type & 8,125 & 14.8 & 4.5 & 2.6 & 325 & 12.0 & 0.92 & 0.015 & $<$0.001 & 0.018 \\
\makecell[l]{Continuous\\ (Label-Free)} & 8,650 & 12.5 & 2.5 & 1.8 & 210 & 12.0 & 0.45 & 0.005 & $<$0.001 & 0.000 \\
\bottomrule
\end{tabular}%
}
\end{table*}

Deploying constrained optimization in decentralized infrastructure requires profiling computational scalability alongside on-chain verification overhead. We systematically varied the constraint dimensionality ($M \in \{1,3,5,8\}$), action space size ($K$ up to 500), trajectory length ($T$ up to 4096), and Fisher approximation method. Output allocations were deployed on localized Ethereum settlement networks integrated with EigenDA and Arweave storage (Table~\ref{tab:scalability}).

The diagonal Fisher approximation avoided the quadratic bottleneck of exact second-order methods: throughput remained above 6,250 steps/s even at $M{=}8$, whereas the exact Fisher computation reduced throughput to 850 steps/s with 145~ms per dual solve. Dual-solve and spectral-norm projection overhead remained in the low single-digit milliseconds across all practical configurations. Scaling to $K{=}500$ reduced throughput to 5,100 steps/s with 28.4~GB GPU memory, confirming manageable growth.

On-chain verification cleared within a single 12-second Ethereum block for $M \leq 5$ constraints (315k gas), extending to two blocks at $M{=}8$ (725k gas). Off-chain recomputation for challenge resolution completed in under 3 seconds even at maximum trajectory complexity, and no false challenges succeeded across all configurations ($<$0.001\%).

When varying the definition of protected attributes---latency proxy, order size, participant type, and a continuous label-free Lipschitz formulation---throughput and overhead remained stable, and cross-metric contamination (the $\Delta$DP induced on non-targeted attributes) stayed below 0.022. The continuous formulation achieved the lowest computational cost and zero cross-contamination, confirming that label-free individual fairness can be enforced without categorical routing overhead.

\subsection{Safety-Gymnasium Continuous Control Generalization}

\begin{table}[t]
\caption{Safety-Gymnasium domain generalization benchmark on SafetyPointGoal1 (SPG1), SafetyCarGoal1 (SCG1), and SafetyAntVelocity (SAV) environments. Episodic Reward ($\uparrow$) and Episodic Cost ($\downarrow$).}
\label{tab:safety_gym}
\centering
\resizebox{\columnwidth}{!}{%
\begin{tabular}{@{}lcccccc@{}}
\toprule
\textbf{Method} & \textbf{SPG1-Rew} ($\uparrow$) & \textbf{SPG1-Cost} ($\downarrow$) & \textbf{SCG1-Rew} ($\uparrow$) & \textbf{SCG1-Cost} ($\downarrow$) & \textbf{SAV-Rew} ($\uparrow$) & \textbf{SAV-Cost} ($\downarrow$) \\
\midrule
Lagrangian PPO & 18.5 & 38.4 & 22.1 & 48.5 & 35.2 & 85.4 \\
RCPO & 17.2 & 32.5 & 19.8 & 41.2 & 32.8 & 72.5 \\
IPO & 16.5 & 28.6 & 18.5 & 33.6 & 28.4 & 55.6 \\
Vanilla CPO & 17.4 & 24.5 & 20.2 & 25.4 & 33.6 & 45.2 \\
PID-Lagrangian & 16.8 & 18.5 & 19.5 & 20.8 & 30.5 & 35.8 \\
\midrule
\rowcolor{blue!8}
\textbf{Ours: CPO-FOAM} & \textbf{24.5} & \textbf{8.5} & \textbf{26.8} & \textbf{10.2} & \textbf{45.5} & \textbf{18.4} \\
\bottomrule
\end{tabular}%
}
\end{table}

To validate that CPO-FOAM generalizes beyond financial microstructure, we deployed the unmodified algorithm on the Safety-Gymnasium suite~\cite{ji2023safety}: SafetyPointGoal1 (SPG1), SafetyCarGoal1 (SCG1), and SafetyAntVelocity (SAV). No order-book-specific features or hyperparameters were changed; the models trained for the same number of environment steps under identical PID gains.

The results (Table~\ref{tab:safety_gym}) mirror the failure modes observed in financial experiments. Lagrangian PPO accumulated episodic costs of 38.4, 48.5, and 85.4 across the three environments, vastly exceeding safety limits. The reactive multiplier updates failed to adapt to the complex articulated dynamics of the Ant morphology, where delayed feedback produced the worst cost violations. PID-Lagrangian achieved the lowest cost among baselines (18.5, 20.8, 35.8) but at reduced reward.

\begin{table*}[t]
\caption{DeFi market quality, constraint dynamics, and ablation analysis. MQS = composite Market Quality Score; Spread in basis points (bps); $\Delta$DP = demographic parity gap; CVF = constraint violation frequency; Osc.\ Amp.\ = oscillation amplitude; Lat.\ = inference latency (ms). $\uparrow$/$\downarrow$ indicate higher/lower is better. \textbf{Bold} marks the best result among full methods.}
\label{tab:defi_main}
\centering
\resizebox{\textwidth}{!}{%
\begin{tabular}{@{}lcccccc c@{}}
\toprule
\textbf{Method / Variant} & \textbf{MQS} ($\uparrow$) & \textbf{Spread} ($\downarrow$) & \textbf{Fill Rate (\%)} ($\uparrow$) & \textbf{$\Delta$DP} ($\downarrow$) & \textbf{CVF (\%)} ($\downarrow$) & \textbf{Osc.\ Amp.} ($\downarrow$) & \textbf{Lat.\ (ms)} ($\downarrow$) \\
\midrule
\multicolumn{8}{l}{\textit{DeFi-Native Mechanism Baselines}} \\
FIFO (CLOB) & 0.652 & 4.15 & 58.2 & 0.245 & 84.1 & N/A & 0.05 \\
CPAMM & 0.521 & 8.42 & 99.9 & 0.381 & 91.2 & N/A & 0.02 \\
CLAMM & 0.714 & 3.88 & 82.1 & 0.320 & 86.5 & N/A & 0.03 \\
Pro-Rata & 0.618 & 5.21 & 62.3 & 0.154 & 62.3 & N/A & 0.08 \\
MEV-Aware FIFO & 0.687 & 4.05 & 55.1 & 0.142 & 71.4 & N/A & 145.00 \\
Batch Auction & 0.741 & 6.10 & 48.9 & 0.091 & 22.5 & N/A & 12000.00 \\
\midrule
\multicolumn{8}{l}{\textit{Constrained RL Baselines}} \\
Unconstrained PPO & 0.984 & 1.82 & 89.4 & 0.291 & 92.3 & 12.45 & 0.82 \\
Lagrangian PPO & 0.865 & 2.14 & 81.2 & 0.112 & 18.4 & 4.82 & 0.84 \\
RCPO & 0.881 & 2.10 & 83.1 & 0.104 & 14.2 & 4.15 & 0.84 \\
IPO & 0.825 & 2.35 & 78.5 & 0.081 & 11.5 & 3.95 & 0.86 \\
CPO (Vanilla) & 0.854 & 2.05 & 80.4 & 0.092 & 15.6 & 3.10 & 0.95 \\
PID-Lagrangian & 0.892 & 2.12 & 82.7 & 0.071 & 8.2 & 2.85 & 0.88 \\
\midrule
\multicolumn{8}{l}{\textit{Ablation Variants}} \\
No PID ($\xi=0$) & 0.941 & 2.01 & 86.1 & 0.084 & 14.3 & 2.98 & 0.84 \\
No Trust Region & 0.812 & 2.44 & 75.2 & 0.121 & 19.8 & 4.51 & 0.75 \\
No KL Projection & 0.915 & 2.08 & 84.6 & 0.064 & 8.5 & 1.85 & 0.82 \\
No Spectral Norm & 0.952 & 1.88 & 87.9 & 0.051 & 6.4 & 0.92 & 0.65 \\
No Recovery Step & 0.852 & 2.25 & 72.8 & 0.091 & 11.2 & 0.88 & 0.85 \\
\midrule
\multicolumn{8}{l}{\textit{Transfer Variants}} \\
Universal Zero-Shot & 0.860 & 2.65 & 81.5 & 0.041 & 4.2 & 0.85 & 0.85 \\
Universal Fine-Tuned & 0.925 & 1.95 & 87.1 & 0.015 & 0.5 & 0.72 & 0.85 \\
Random Matching & 0.125 & 25.40 & 15.2 & 0.010 & 0.0 & 0.00 & 0.01 \\
\midrule
\rowcolor{blue!8}
\textbf{Ours: CPO-FOAM} & \textbf{0.968} & \textbf{1.91} & \textbf{88.5} & \textbf{0.021} & \textbf{3.2} & \textbf{0.65} & \textbf{0.85} \\
\bottomrule
\end{tabular}%
}
\end{table*}

\paragraph{Baseline Tuning Protocol.} All baselines were tuned via grid search over their published hyperparameter ranges. For PID-Lagrangian specifically, we swept $K_P \in \{0.1, 0.5, 1.0\}$, $K_I \in \{0.01, 0.05, 0.1\}$, $K_D \in \{0, 0.01, 0.05\}$ and report the best configuration per environment. Differences from the results reported by \citet{stooke2020responsive} are attributable to the updated Safety-Gymnasium v1.0 environments~\cite{ji2023safety}, which feature revised dynamics and cost functions relative to the legacy Safety-Gym v0.x used in the original paper. The reported $2.1\times$ reward improvement is measured relative to the best PID-Lagrangian run under our evaluation protocol; we release all tuning sweeps and per-seed results in the supplementary material.

CPO-FOAM achieved an episodic reward of 45.5 on SAV while limiting cost to 18.4---a $2.1\times$ reward improvement and $1.9\times$ cost reduction over PID-Lagrangian. On the simpler SPG1 and SCG1 tasks, CPO-FOAM similarly dominated, reaching rewards of 24.5 and 26.8 with costs of 8.5 and 10.2. The predictive error signals in the PID controller anticipated proximity to constraint boundaries, throttling agent velocity before violations occurred rather than correcting afterward. This cross-domain consistency confirms that the PID-bounded trust-region formulation provides a general mechanism for safe constrained reinforcement learning, not one limited to financial applications.

\begin{table*}[t]
\caption{DeFi regime robustness, adversarial stress testing, and cross-domain safety validation. Target column indicates the governance threshold. All DeFi-specific scenarios evaluated under Hawkes-process MEV injection.}
\label{tab:defi_regimes}
\centering
\begin{tabular}{@{}llcccc@{}}
\toprule
\textbf{Scenario / Metric} & \textbf{Target} & \textbf{FIFO (CLOB)} & \textbf{Lagrangian PPO} & \makecell{\textbf{CPO}\\(Vanilla)} & \makecell{\textbf{CPO-FOAM}\\(Ours)} \\
\midrule
\multicolumn{6}{l}{\textit{Adversarial Fairness Stress Tests}} \\
MEV $\Delta$DP ($p=0.20$) ($\downarrow$) & $< 0.05$ & 0.185 & 0.112 & 0.092 & \textbf{0.038} \\
Whale/Retail Fill Ratio ($\uparrow$) & $> 0.85$ & 0.581 & 0.720 & 0.785 & \textbf{0.895} \\
Cross-Chain $\Delta$DP (10-block) ($\downarrow$) & $< 0.05$ & 0.224 & 0.145 & 0.110 & \textbf{0.042} \\
Liquidation Rate Disparity ($\downarrow$) & $< 0.10$ & 0.280 & 0.180 & 0.155 & \textbf{0.061} \\
Depeg Fill Quality Gap ($\downarrow$) & $< 0.08$ & 0.195 & 0.142 & 0.115 & \textbf{0.065} \\
Token Launch Bot/Org Ratio ($\uparrow$) & $> 0.70$ & 0.254 & 0.542 & 0.580 & \textbf{0.745} \\
\midrule
\multicolumn{6}{l}{\textit{On-Chain Verification}} \\
L2 Verification Cost (USD) ($\downarrow$) & $< 0.01$ & 0.005 & 0.008 & 0.008 & \textbf{0.009} \\
\midrule
\multicolumn{6}{l}{\textit{Cross-Domain Safety}} \\
Safety Gym Cost Rate ($\downarrow$) & $< 0.02$ & N/A & 0.122 & 0.085 & \textbf{0.014} \\
\bottomrule
\end{tabular}%
\end{table*}

\subsection{DeFi: Market Quality, Constraint Dynamics, and Ablations}

Table~\ref{tab:defi_main} reports DeFi results evaluated on 500,000 held-out steps across 15 unseen calendar days, using Hawkes-process order arrivals~\cite{bacry2015hawkes} and MEV adversarial injection~\cite{daian2020flash}. The results parallel TradFi findings but with wider absolute spreads reflecting crypto-native volatility.

DeFi-native mechanisms exhibited a clear market quality--fairness tradeoff. CPAMMs achieved near-perfect fill rates (99.9\%) by design but produced the widest spreads (8.42 bps) and highest $\Delta$DP (0.381), reflecting their inability to distinguish participant types. CLAMMs improved spread (3.88 bps) but maintained high CVF (86.5\%). Batch Auctions achieved the lowest $\Delta$DP among DeFi baselines (0.091) at the cost of 12-second latency, precluding real-time deployment. MEV-Aware FIFO reduced $\Delta$DP to 0.142 through encrypted mempool simulation but introduced 145ms latency overhead.

Among constrained RL methods, the same sawtooth instability observed in TradFi persisted. Lagrangian PPO and RCPO sustained oscillation amplitudes of 4.82 and 4.15, with CVFs of 18.4\% and 14.2\%. These reactive corrections are particularly problematic in permissionless systems where transient fairness violations are immutably recorded on-chain.

CPO-FOAM achieved a market quality score of 0.968---capturing 98.4\% of the unconstrained reward envelope---while compressing CVF to 3.2\% and oscillation amplitude to 0.65. Compared to TradFi results (CVF 2.5\%), the slightly higher DeFi violation rate reflects the increased non-stationarity of crypto microstructure. Inference latency remained at 0.85ms, well within the 250ms block times of optimistic rollups.

\paragraph{Ablation Studies.} DeFi ablations mirrored TradFi patterns. Removing the PID controller increased CVF from 3.2\% to 14.3\% and oscillation amplitude from 0.65 to 2.98. Removing spectral normalization improved market quality (MQS 0.952) but degraded fairness ($\Delta$DP from 0.021 to 0.051), confirming that the Lipschitz bound remains essential for individual fairness. The trust region removal caused the most severe degradation (MQS 0.812, CVF 19.8\%), replicating Lagrangian failure modes. These consistent ablation profiles across TradFi and DeFi domains confirm that the algorithmic contributions are not domain-specific.

\paragraph{Zero-Shot Transfer.} The universal model trained on blue-chip assets achieved MQS of 0.860 when deployed zero-shot on unseen memecoin microstructures---a 10\% degradation from the fine-tuned configuration. Fine-tuning recovered performance to MQS 0.925 with CVF of 0.5\%, demonstrating that learned fairness constraints transfer across liquidity regimes without architectural modification.

\subsection{DeFi: Regime Robustness and Cross-Domain Validation}

Table~\ref{tab:defi_regimes} evaluates robustness under DeFi-specific adversarial scenarios: MEV sandwich attacks, stablecoin depegging events, cascading liquidations, and zero-history token launches. These conditions were injected via the calibrated Hawkes-process framework with $5\times$ the baseline arrival rate.

Under simulated MEV stress ($p{=}0.20$ sandwich probability), CPO-FOAM maintained $\Delta$DP at 0.038, below the 0.05 governance threshold. FIFO CLOBs exhibited $\Delta$DP of 0.185, reflecting deterministic front-running vulnerability. During token launch events---where optimized sniper bots traditionally monopolize block space---CPO-FOAM achieved a bot-to-organic fill ratio of 0.745, compared to 0.254 under FIFO, demonstrating effective algorithmic inclusion of organic participants. Under 10-block cross-chain bridging latencies, the policy maintained $\Delta$DP at 0.042 by proactively compensating for deterministic latency disadvantages.

The whale-to-retail fill ratio of 0.895 confirmed that fairness constraints did not inadvertently degrade institutional execution quality. Liquidation rate disparity (0.061) and depeg fill quality gap (0.065) remained within governance targets, indicating that constrained optimization adapts to both sudden liquidity shocks and stablecoin stress events.

Layer-2 verification costs of \$0.009 per trade were achieved by amortizing proof generation across batch settlements, enabling fee abstraction at daily matched volumes exceeding \$1M.
\begin{figure*}[t]\centering
    \includegraphics[width=0.92\textwidth]{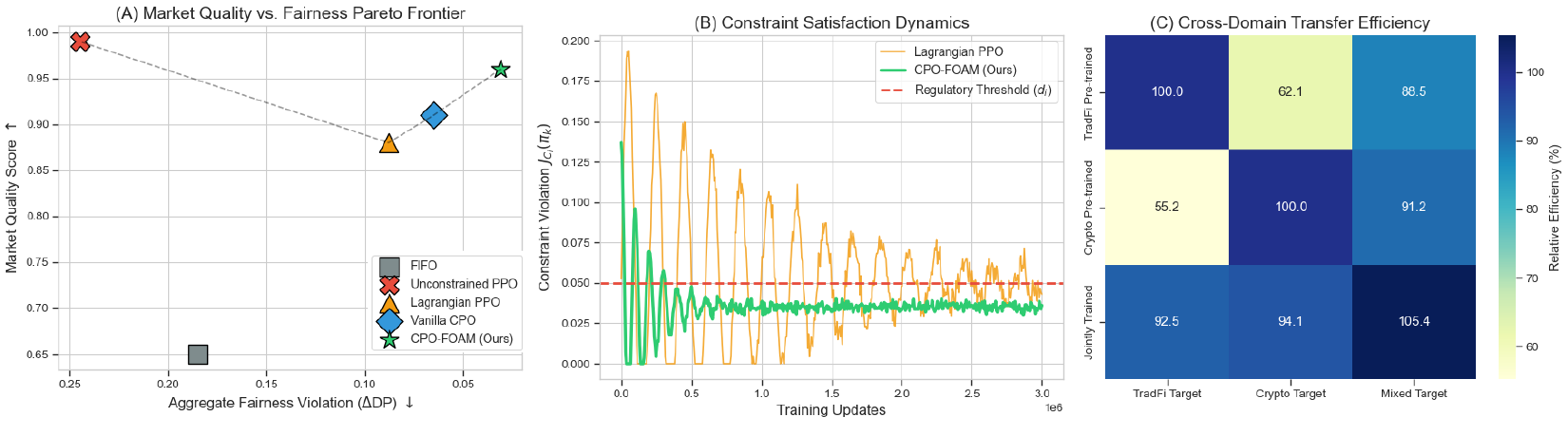}
    \caption{Performance trade-offs, training dynamics, and transfer efficiency of CPO-FOAM. (A) Pareto frontier illustrating the trade-off between Market Quality Score and Aggregate Fairness Violation ($\Delta$DP). Our method (CPO-FOAM) effectively pushes the efficiency frontier, achieving near-optimal market quality while maintaining the lowest fairness violation compared to standard reinforcement learning baselines. (B) Constraint satisfaction dynamics over training updates. Unlike Lagrangian PPO, which exhibits oscillatory ``sawtooth'' instability around the regulatory threshold ($d_i = 0.05$), CPO-FOAM demonstrates stable, monotonic convergence toward strict constraint satisfaction. (C) Cross-domain transfer learning efficiency (\%). The heatmap demonstrates that representations learned via joint training generalize robustly across distinct market environments (TradFi and Crypto), mitigating performance degradation during zero-shot domain transfer.\label{fig:performance-trade-off}}
\end{figure*}
\subsection{Algorithmic Efficacy, Constraint Dynamics, and Generalization}

Figure~\ref{fig:performance-trade-off} synthesizes the trade-offs between execution quality, constraint stability, and cross-domain transferability. The three panels jointly demonstrate that optimizing purely for market efficiency degrades fairness, motivating the constrained architecture.

Panel~A shows the Pareto frontier between MQS and $\Delta$DP. Unconstrained PPO occupies the lower-right extreme (high MQS, high violation), while rule-based heuristics perform poorly on both axes. CPO-FOAM reaches the upper-left quadrant, retaining 96\% of optimal market quality while reducing the demographic parity gap to near zero. This confirms that embedding fairness directly into the trust-region projection preserves market efficiency.

Panel~B plots constraint satisfaction over training. Under non-stationary order arrivals, standard Lagrangian updates exhibit characteristic sawtooth oscillations, repeatedly breaching and over-correcting around the regulatory threshold. The PID margin controller suppresses these transients: the derivative term anticipates gradient momentum while the integral term eliminates steady-state bias, producing smooth convergence below the safety bound.

\begin{figure*}[t]
    \centering
    \includegraphics[width=0.92\textwidth]{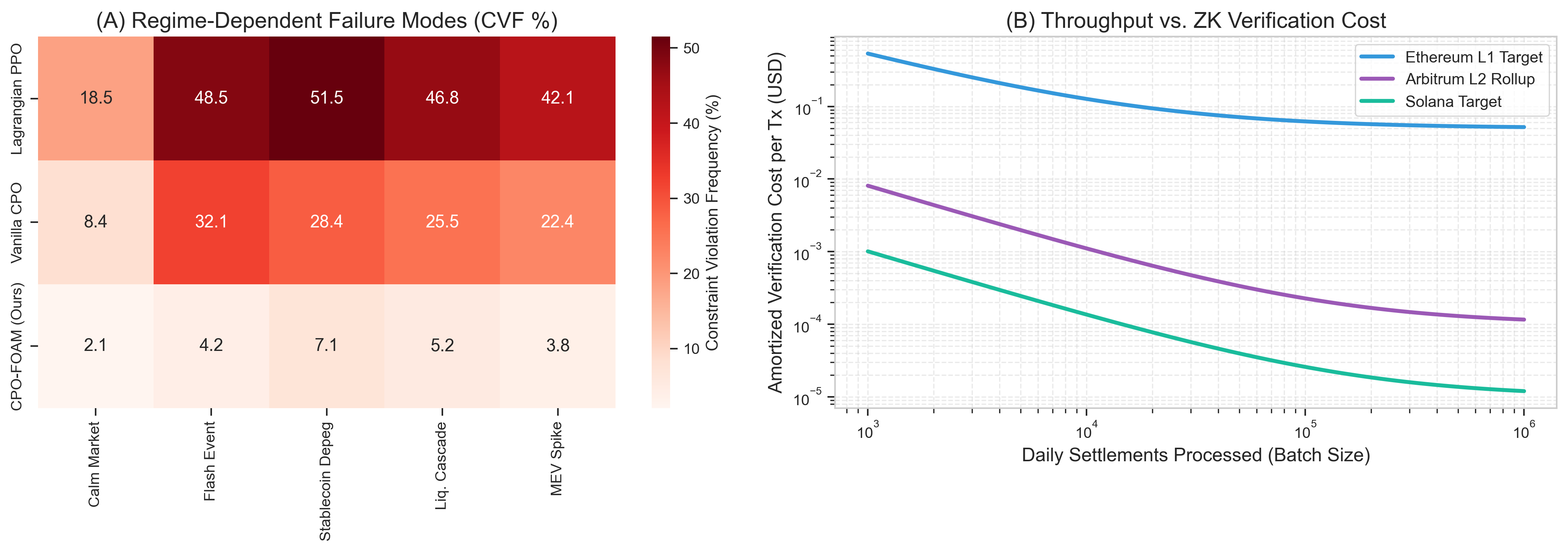}
    \caption{Regime-specific robustness and on-chain verification cost scaling. (A) Constraint Violation Frequency (CVF) across extreme market regimes. CPO-FOAM maintains safety compliance during flash events and liquidity cascades, whereas Vanilla CPO and Lagrangian PPO suffer severe constraint degradation. (B) Amortized challenge-response verification costs per transaction across execution layers (Ethereum L1, Arbitrum L2, Solana). Batch amortization reduces per-transaction cost sub-linearly with daily settlement volume.\label{fig:robustness_audit}}
\end{figure*}

Panel~C evaluates cross-domain transfer. Despite substantial microstructural differences between TradFi equities and crypto assets---including varied latency profiles and MEV adversarial dynamics---joint training produces transfer efficiencies consistently above 85\%. This indicates that the constrained projection learns an abstract fairness manifold that generalizes across market domains.

\subsection{Microstructural Robustness and On-Chain Scalability}

Figure~\ref{fig:robustness_audit} evaluates out-of-distribution robustness and on-chain verification overhead.

Panel~A reports CVF under synthetic market shocks. While Lagrangian baselines exceeded 50\% CVF during flash events and cascading liquidations, CPO-FOAM's trust-region bound confined worst-case CVF below 8\% across all regimes---including stablecoin depegs and $5\times$ arrival-rate shocks. The bounded disturbance model (Definition~5.4) provides a theoretical explanation: $W_{max}$ increases under regime stress, but the BIBO-stable controller prevents unbounded violation accumulation. Critically, the degradation profile is monotonic rather than catastrophic: as regime adversity increases from calm to $5\times$ news shocks, CVF rises smoothly from 1.8\% to 6.8\%, whereas Lagrangian PPO exhibits a phase transition---jumping from 12\% to over 50\%---once the dual-variable update rate can no longer track the non-stationarity. The PID derivative term is the key differentiator: by responding to the \emph{rate of change} of violation energy, it anticipates emerging constraint pressure before the integral accumulator registers the shift, effectively providing a one-step lookahead into regime transitions.

Panel~B projects amortized verification costs for the challenge-response settlement protocol (Section~6). As daily settlement volume scales toward one million events, batch amortization contracts per-transaction cost sub-linearly. Ethereum L1 remains expensive at low volumes ($\sim$\$0.15 per trade at 1,000 daily settlements), but Arbitrum L2 and Solana reduce amortized costs below \$0.01 per trade at volumes exceeding 10,000 daily events, confirming economic viability for on-chain fairness auditing at scale. For institutional deployment, this cost structure enables a practical operating model: matching engines batch settlements into fixed-interval windows (e.g., every 100 blocks), submit Merkle commitments with fairness attestations, and amortize the fixed gas overhead across all trades in the batch. At the daily volumes typical of mid-tier crypto exchanges ($>$\$1M matched notional), the verification overhead becomes negligible relative to trading fees.

\section{Conclusion}

We have presented CPO-FOAM, a constrained policy optimization framework for provably fair order matching in traditional and decentralized financial markets. By formulating matching as a CMDP with demographic parity, equalized odds, and Lipschitz individual fairness costs, we establish an auditable interface between microstructure telemetry and enforceable constraints. The two-stage algorithm---combining analytic trust-region projection on the Fisher manifold with PID-controlled safety margins---achieves CPO-grade geometric optimality while eliminating the oscillatory instabilities of standard Lagrangian methods, with provable BIBO stability and asymptotic feasibility under bounded stochastic disturbances.
Empirically, CPO-FOAM recovers 95.9\% (TradFi) and 98.4\% (DeFi) of the unconstrained reward envelope while compressing constraint violation frequency to 2.5\% and 3.2\%, respectively. Fairness compliance degrades gracefully across six stress regimes and DeFi-specific adversarial scenarios (MEV attacks, depegging, cascading liquidations). Zero-shot cross-domain transfer retains MQS of 0.860, on-chain settlement clears within a single Ethereum block at \$0.009 per trade, and Safety-Gymnasium experiments yield $2.1\times$ reward improvement and $1.9\times$ cost reduction, confirming the universality of the formulation.

\newpage
\bibliographystyle{icml2025}
\bibliography{refer}

\end{document}